\RequirePackage{amsmath}
\documentclass[runningheads]{llncs}
\usepackage[T1]{fontenc}
\usepackage{float}
\usepackage{graphicx}
\usepackage{subcaption}
\usepackage[inline]{enumitem}
\usepackage{tikz} 
\usetikzlibrary{
    calc,
    decorations.pathmorphing,
    decorations.pathreplacing, 
    shapes.geometric, 
} 

\usepackage{macros}
\usepackage[capitalise,noabbrev]{cleveref}

\begin{document}
\title{Realizing Graphs with Cut Constraints}
%

\author{
Vítor Gomes Chagas\inst{1}\orcidID{0000-0002-6506-4174}
\and \\
Samuel Plaça de Paula\inst{1}\orcidID{0009-0005-5970-2984}
\and \\
Greis Yvet Oropeza Quesquén\inst{1}\orcidID{0000-0003-0112-8009}
\and \\
Lucas de Oliveira Silva\inst{1}\orcidID{0000-0002-7846-5903}
\and \\
Uéverton dos Santos Souza\inst{2,3}\orcidID{0000-0002-5320-9209}
}

\authorrunning{V. G. Chagas et al.}
%

\institute{
Instituto de Computação, Universidade Estadual de Campinas, Campinas, Brazil\\
\email{vitor.chagas@ic.unicamp.br} \\
\email{s233554@dac.unicamp.br} \\
\email{greis.quesquen@ic.unicamp.br} \\
\email{lucas.oliveira.silva@ic.unicamp.br} \\
\and
IMPA, Instituto de Matemática Pura e Aplicada, Rio de Janeiro, Brazil\\
\and
Instituto de Computação, Universidade Federal Fluminense, Niterói, Brazil\\
\email{ueverton.souza@impatech.org.br}}

\maketitle 
\begin{abstract}

Given a finite non-decreasing sequence $\texttt{d}=(d_1,\ldots,d_n)$ of natural numbers, 
the \GRfull~problem asks whether \texttt{d} is a graphic sequence, i.e., there exists a labeled simple graph such that $(d_1,\ldots,d_n)$ is the degree sequence of this graph. Such a problem can be solved in polynomial time due to the Erd\H{o}s and Gallai characterization of graphic sequences. 
Since vertex degree is the size of a trivial edge cut, we consider a natural generalization of \GRfull, where we are given a finite sequence $\texttt{d}=(d_1,\ldots,d_n)$ of natural numbers (representing the trivial edge cut sizes) and a list of nontrivial cut constraints $\call$ composed of pairs $(S_j,\ell_j)$ where $S_j\subset \{v_1,\ldots,v_n\}$, and $\ell_j$ is a natural number.
In such a problem, we are asked whether there is a simple graph with vertex set $V=\{v_1,\ldots,v_n\}$ such that $v_i$ has degree $d_i$ and $\partial(S_j)$ is an edge cut of size $\ell_j$, for each $(S_j,\ell_j)\in \call$. We show that such a problem is polynomial-time solvable whenever each $S_j$ has size at most three. Conversely, assuming P~$\neq$~NP, we prove that it cannot be solved in polynomial time when $\call$ contains pairs with sets of size four, and our hardness result holds even assuming that each $d_i$ of \texttt{d} equals $1$. 

\end{abstract}
%
%
%
\newpage

\section{Introduction}
\label{sec:introduction}

Graph realization is a fundamental combinatorial problem in the field of Graph Theory, and its studies have fostered interest in and understanding of the discrete structure of graphs. Nowadays, graph realization is a topic commonly covered in introductory Graph Theory courses, providing those new to the world of graphs with many insights into their combinatorial properties.
The Handshaking Lemma, for example, is usually one of the first statements that beginners come across when they begin studying graphs. Although simple to understand, it is the gateway to a world of more intriguing graph-related questions. By observing that the sum of the degrees of all vertices is equal to twice the number of edges in the graph, it follows that not every sequence of $n$ natural numbers can be the degree sequence of some graph with $n$ vertices, and it becomes natural to ask when a sequence of $n$ natural numbers admits a graph with $n$ vertices whose degree sequence corresponds to the given sequence; the topic related to such a question is called \emph{graph realization}.

Given a sequence $\texttt{d}$ of $n$ natural numbers that satisfy the Handshaking Lemma (i.e., the sum of its values is even), it is a simple exercise to verify that it is possible to construct a multigraph (parallel edges and loops are allowed) whose degree sequence corresponds to $\texttt{d}$. However, this question becomes more intriguing when the goal is to realize a simple graph where parallel edges and loops are not allowed. A non-decreasing sequence $\texttt{d}=(d_1,\ldots,d_n)$ of natural numbers is said to be \emph{graphic} if it is realizable by a simple graph, that is, if there exists a labeled simple graph $G$ with $n$ vertices such that \texttt{d} is its degree sequence.
Formally, the classical {\sc Graph Realization} problem is stated as follows:

\defproblema{Graph Realization}
{A non-decreasing sequence $\texttt{d} = (d_1, \dots, d_n)$ of natural numbers.}
{Is $\texttt{d}$ a graphic sequence?}

In 1960, Erd\H{o}s and Gallai provided necessary and sufficient conditions for a sequence of non-negative integers to be graphic, proving the following theorem.

\begin{theorem}[Erdős and Gallai~\cite{erdos60}]\label{erdosgallaithm}
    A non-decreasing sequence $\texttt{d} = (d_1, \dots, d_n)$ of natural numbers is graphic if and only if \\
    \begin{enumerate*}
        \item $\sum\limits_{i=1}^n d_i$ is even, and %
        \item $\sum\limits_{i=1}^k d_i \le k(k - 1) + \sum\limits_{i=k+1}^n \min\{d_i, k\}$, for every $1 \le k \le n$.
    \end{enumerate*}
\end{theorem}

It is not difficult to see that the conditions presented by Erd\H{o}s and Gallai are necessary. However, the sufficiency proof provided by Erd\H{o}s and Gallai is quite elaborated. Several alternative proofs for this sufficiency condition have been shown since then until recently, such as Harary~\cite{harary1969graph} in 1969, Berge~\cite{berge1973graphs} in 1973,  Choudum~\cite{choudum1986simple} in 1986, Aigner and Triesch~\cite{aigner1994realizability} in 1994, Tripathi and Tyagi~\cite{tripathi2008simple} in 2008, and Tripathi, Venugopalan, and West~\cite{tripathi2010short} in 2010. 
From Theorem~\ref{erdosgallaithm}, it follows that \GRfull~ can solved in polynomial time. Additionally, Tripathi, Venugopalan, and West~\cite{tripathi2010short} presented a simple constructive proof of Theorem~\ref{erdosgallaithm} that allows us to obtain the graph to be realized in $\mathcal{O}(n\cdot \sum_{i=1}^n d_i)$ time.
Furthermore, algorithms like the one provided by Havel and Hakimi~\cite{Ha55,Ha62}, which iteratively reduce the degree sequence while maintaining its realizability, also give a constructive approach to finding such graphs if one exists. Havel and Hakimi's algorithms run in
$\mathcal{O}(\sum_{i=1}^n d_i)$ time, which is optimal.

Variants of the \GRfull{} problem requiring that the realizing graph belongs to a particular graph class have also been studied in the literature. Examples of already studied classes included trees~\cite{gupta2007graphic}, 
Split graphs~\mbox{\cite{hammer1981splittance,chat2014recognition}}, Chordal, interval, and perfect graphs~\cite{chernyak1987forcibly}. Besides that, sequence pairs representing the degree sequences of a bipartition in a realizing bipartite graph were also studied in~\cite{burstein2017sufficient}. Surprisingly, the question regarding \GRfull{} for the class of bipartite graphs appears to remain open for over 40 years~\mbox{\cite{bar2022realizing,rao2006survey}}. 
In addition, the problem of determining whether a given sequence defines a unique realizing simple graph was studied in~\cite{aigner1994realizability,KOREN1976235,pak2013constructing}, and Bar-Noy, Peleg, and Rawitz~\cite{bar2020vertex} introduced the vertex-weighted variant of \GRfull{} where we are given a sequence $\texttt{d}=(d_1,d_2,\ldots,d_n)$ representing a ``weighted degree'' sequence, and a vector $\texttt{w}=(w_1,w_2,\ldots,w_n)$ representing vertex weights, and asked whether there is a graph with vertex set $V=\{v_1,v_2,\ldots,v_n\}$ such that for each $v_i$ the sum of the weights of its neighbors is equal to $d_i$. 

In today's interconnected world, many fields face the challenge of structuring systems with specific connectivity requirements.
For instance, in social network analysis, building a network where individuals (vertices) have a fixed number of connections (degree) is crucial for analyzing influence, community structures, and information diffusion. 
Similarly, urban planners confront similar issues when designing road networks, where intersections must be connected with a specific number of roads to optimize traffic flow. 
These examples highlight the significance of addressing connectivity challenges in various fields where specific connectivity patterns must be achieved. 
One of the most studied problems in this context is the realization problem that deals with degree sequences.  
According to Bar-Noy, Böhnlein, Peleg, and Rawitz~\cite{bar2022vertex}, \GRfull{} and its variants have interesting applications in network design, randomized algorithms, analysis of social networks, and chemical networks.

In this paper, in the same flavor as Bar-Noy, Peleg, and Rawitz~\cite{bar2020vertex}, we introduce another natural variant of the \GRfull{}, which we propose calling \textsc{Graph Realization with Cut Constraints}. 
First, we start with some definitions. For a vertex set $V = \srange{v_1}{v_n}$, a cut list is defined as a list of pairs $\call = \{(S_1, \ell_1), \dots, (S_m, \ell_m)\}$, where each pair $(S_j, \ell_j) \in \mathcal{L}$ consists of a nonempty, proper subset $S_j \subset V$ and a natural number $\ell_j$.
Given a cut list $\call$ for a set $V$ and a graph $G$ with vertex set $V$, we say that $G$ \emph{realizes} $\mathcal{L}$ if, for every pair $(S_j, \ell_j) \in \mathcal{L}$, the edge cut $\partial(S_j)$ has size $\ell_j$. For a cut list $\mathcal{L}$, we denote by $w(\call) = \max\limits_j\ \size{S_j}$ the largest size among the subsets $S_j$ in $\mathcal{L}$. 
Now, we define our problem:

\defproblema{Graph Realization with Cut Constraints (\GRC{})}
{A cut list $\call$ for a set of vertices $V=\srange{v_1}{v_n}$, and a non-decreasing sequence $\texttt{d} = (d_1, \dots, d_n)$ of natural numbers.}
{Does there exist a simple graph with vertex set $V$ such that, for every $i$, $v_i$ has degree $d_i$ and $G$ realizes $\call$?}

Recall that the degree of a vertex $v_i$ of a graph $G$ is the size of the trivial edge cut $\partial(\{v_i\})$ in $G$. Therefore, one can see \GRfull{} as given a cut list \texttt{d} of all trivial edge cut sizes $(\{v_i\},d_i)$, decide whether there is a graph $G$ realizing \texttt{d}.
In \GRC{}, we assume that  $\call$ is a list of some nontrivial edge cut sizes for the realizing graph, i.e., each $S_j$ has a size of at least two and at most $n-2$.
If~$\call=\emptyset$, then the problem becomes the original \GRfull~problem. So, through this work, we always consider $\call\neq \emptyset$ and $w(\call)\geq 2$.

%
Although our problem, as we have defined, has never been explored before, it was motivated by an active research topic with several recent results \cite{Ap22,Li24} that aims to learn an unknown graph $G$ or properties of $G$ via cut-queries.
In this context, given a graph $G = (V, E)$ with a known vertex set but an unknown edge set, the objective is to reconstruct $G$ or compute some property of $G$ with a minimal number of queries. A cut-query receives $S\subseteq V$ as input and returns the size of the edge cut $\partial(S)$.
One of the main driving interests in this model is its connection with submodular function minimization \cite{Aa05}. Furthermore, these active learning questions have applications in fields like computational biology \cite{Vlad98} and relate to data summarization, where queries reveal ``relevant information'' about the graph. More generally, this type of question can be viewed as a means of determining a property of an unknown object via indirect queries about it \cite{Ai88,Du00}.

Within this framework, our problem can be viewed as a validity check to test whether the cut queries are consistent and if there is some graph that satisfies them. Also, it can be viewed as a variant, where we cannot choose the queries, but rather, we are given cut constraints and want to find one satisfying candidate graph. 
Concerning cut-queries, knowing $\partial(\set{u}), \partial(\set{v}),$ and $\partial(\set{u,v})$, we find out whether or not there is an edge between $u$ and $v$ in $G$. Thus, with at most $\binom{n}{2}+n$ queries, one can always obtain the edge set of  $G$. Similarly, regarding \GRC{}, if $\call$ contains all possible sets of size two, then the problem is trivial. Therefore, in this paper, we are mainly interested in the case where $\call$ has polynomial size with respect to $n$ and does not contain all sets of size two.

\paragraph{Our Contribution.} 
In this work, we study the \GRC{} problem and provide a comprehensive characterization of its computational complexity, focusing on the size of the cut sets involved. 
%
We show that it is polynomial-time solvable for instances where $w(\call)\le3$. Specifically, when $w(\call)=2$, the problem reduces to the classic $f$-factor problem,  which can be solved in polynomial time. 
Additionally, we show that cuts of size three, surprisingly, do not increase complexity. Instances involving such cuts can be transformed into equivalent ones where size-three constraints are replaced by cut constraints involving only sets of size at most two, all while preserving the same realizability. 
This shows that even with cut constraints using sets of size three, the problem remains polynomial-time solvable. 
%
On the other hand, we also prove that when cut sets of size four or larger are allowed, the problem cannot be solved in polynomial time unless P~$=$~NP. 
This provides a complete dichotomy regarding the computational complexity of the problem and the size of the cut sets.
In addition, we also prove the NP-completeness for $w(\call)= 6$ when the cut constraints restrict the possibility graph, formally defined in Section~\ref{sec:preliminaries}, to be subcubic and bipartite. 
In contrast, when the cut constraints restrict the possibility graph to be a tree, the problem is solvable in polynomial time. 
%

\paragraph{Related work.}
Several other generalizations and related problems exist in the study of degree sequences and graph realizability. 
Aigner and Triesch~\cite{aigner1994realizability} explored the realizability and uniqueness of graphs based on two types of invariants (degree sequences and induced subgraph sizes), focusing on both directed and undirected graphs and their computational complexity.
Similarly, Erdős and Miklós~\cite{erdos2018} discussed the complexity of degree sequence problems, focusing on the second-order degree sequence problem, which is shown to be strongly \classNPC{}. 
Erdős et al.~\cite{erdos17} presented a skeleton graph structure for a more general restricted degree sequence problem, studying two cases with specific edge restrictions and examining the connectivity of the realization space.
Iványi~\cite{Iv12} explored conditions and algorithms for determining if a sequence is the degree sequence of an $(a, b, n)$-graph,
which is a (directed or undirected) graph whose vertices degrees are in the $[a, b]$ range.

Another field in Graph Theory that is closely related to the \GRfull{}~problem is the study of graph factors and factorizations.
A \textit{factor} of a graph $G$ is simply a spanning subgraph of $G$.
There have been several studies on graph factors under different constraints, such as conditions on their degrees or restrictions on the classes that they must belong to.
Here, we are particularly interested in graph factors described by their degrees, which we call degree factors.
In this context, given an integer $k$, a \emph{$k$-factor} of a graph $G$ is a $k$-regular spanning subgraph of $G$.
This generalizes many problems, for instance a $1$-factor is the same as a perfect matching, and studies in this area date back to the $19$th century when Petersen \cite{petersen1900} gave one of the first sufficient conditions for a $1$-factor.

The concept of $k$-factors has been generalized to consider other values of degrees rather than a fixed number.
Given two functions $g, f \colon V \to \N$ such that $g \leq f$, a spanning subgraph $H$ of the graph $G = (V, E)$ is a \emph{$(g, f)$-factor} if for every vertex $v$, it holds that $g(v) \leq d_H(v) \leq f(v)$.
If $g = f$, then it is simply called an \emph{$f$-factor}.
The problem of determining if a graph admits a $(g, f)$-factor is known to be solvable in polynomial time \cite{Ans85}.
As will be shown later, the \GRC{} problem generalizes the $f$-factor problem.
Classical results in this field include 
Tutte's theorem on $f$-factors \cite{tutte1952} 
and Lovasz's characterization of~\mbox{$(g,f)$-factors}~\cite{Lovasz70}.
For a detailed treatment of this topic, we refer the reader to the surveys of Akiyama and Kano \cite{akiyama1985} and Plummer~\cite{plummer2007}.

\paragraph{Organization of the text.}
The remainder of this paper is organized as follows. 
In \cref{sec:preliminaries}, we introduce key definitions and notations related to the \GRC{} problem, including some conditions for the realizability of a \GRC{} instance. 
%
In \cref{sec:small_cuts}, we investigate the \GRC{} problem with cut sizes restricted to three,
while \cref{sec:large_cuts} focuses on instances with cuts of size at least four.
Finally, in \cref{sec:final_remarks}, we summarize our results and discuss potential extensions of this work. 
Due to space constraints, some proofs have been omitted.


\section{Preliminaries}
\label{sec:preliminaries}

Let $S$ be a set of vertices of a graph $G=(V, E)$.
We denote the total degree of vertices in $S$ by $d(S) = \sum_{u \in S} d_u$, where $d_u$ represents the degree of vertex $u$.
A simple observation is that the size of the edge cut $\partial(S)$ is determined by the degree of the vertices in $S$ and the edges between vertices of $S$.
If there are $k$ edges between vertices of $S$, then $\size{\partial(S)}=d(S)- 2k$.
Since the number of edges in $S$ may vary from $0$ to $\binom{\size{S}}{2}$, a necessary condition for the realizability of an \GRC{} instance is as follows.

\begin{remark}
\label{thm:feasible_cut_sizes}
    A \GRC{} instance $(\texttt{d}, \call)$ is realizable only if, for each cut $(S, \ell) \in \call$, we have $\ell \in \set{ d(S) - 2k : 0 \leq k \leq \binom{\size{S}}{2} }$.
\end{remark}

Since this condition is easily verifiable, we assume henceforth that it holds for any \GRC{} instance. In particular, for cuts of size two, this observation implies that only two feasible values are possible, determining whether an edge must exist between the corresponding vertices, as detailed below.

\begin{remark}
\label{thm:fixed_forbidden_edges}
    Given an instance $I = (\texttt{d}, \call)$ of \GRC{}, in any realization $G$ of $I$, if $(\set{u, v}, d_u + d_v - 2) \in \call$, then $uv \in E(G)$, and if $(\set{u, v}, d_u + d_v) \in \call$, then $uv \notin E(G)$.
\end{remark}

Based on this, we say that an edge $uv$ is \textit{fixed} if $(\set{u, v}, d_u + d_v - 2) \in \call$ and is \textit{forbidden} if $(\set{u, v}, d_u + d_v) \in \call$.
We apply similar terminology when constructing an instance of \GRC{}.
Given an instance $(\texttt{d}, \call)$ of \GRC{}, to \textit{fix} or \textit{forbid} an edge $uv$ means adding the cut $(\set{u, v}, d_u + d_v - 2)$ or $(\set{u, v}, d_u + d_v)$ to $\call$, respectively.

\cref{thm:fixed_forbidden_edges} implies that the \GRC{} problem, when limited to cuts of size two, is equivalent to the \GR{} problem with added constraints: a subset of edges is fixed, and another disjoint one is forbidden. Moreover, we can simplify the problem by focusing only on forbidden edges by reducing the degree of vertices incident to fixed edges and then marking those edges as forbidden.
Formally, given an instance $(\texttt{d}, \call)$ and a cut $(\set{u, v}, d_u + d_v - 2) \in \call$, in which case the edge $uv$ is fixed, we can produce an equivalent instance $(\texttt{d}', \call')$ as follows.  For all $i \notin \{u, v\}$  set $d'_i = d_i$. Reduce $d'_u = d_u - 1$ and   $d'_v = d_v - 1$;  
    %
and $\call'$ is obtained from $\call$ by replacing $(\set{u, v}, d_u + d_v - 2)$ with $(\set{u, v}, d_u + d_v)$.

The resulting instance $(\texttt{d}', \call')$ has a realization if and only if $(\texttt{d}, \call)$ has a realization. If $G = (V, E)$ is a realization of $(\texttt{d}, \call)$, then, as discussed above, we must have $uv \in E$, and $G - uv$ is a realization of $\call'$. Conversely, if $G' = (V, E')$ is a realization of $(\texttt{d}', \call')$, then necessarily $uv \notin E'$ due to the cut $(\set{u, v}, d_u + d_v)$, and $G' + uv$ is a realization of $(\texttt{d}, \call)$.

Thus, cut restrictions involving sets of size two can be simply reinterpreted as forbidding edges.
Let $F$ be the set of all forbidden edges that cannot appear in any realization of instance $(\texttt{d}, \call)$. Then $\calg = K_n - F$ is what we call the \emph{possibility graph}, which must be a supergraph of any valid realization of~$(\texttt{d}, \call)$.


\section{Small Cuts}
\label{sec:small_cuts}

In this section, we show that the \GRC{} problem can be solved in polynomial time for instances~$(\texttt{d}, \call)$ where $w(\call) \leq 3$.
Reinterpreting the size-two cuts of $\call$ as forbidden edges allows us to transform the \GRC{} problem into an equivalent formulation of the classic $f$-factor problem whenever $w(\call) = 2$. This leads us to the following conclusion.

\begin{lemma}
\label{thm:size2}
    Any instance $I=(\texttt{d}, \call)$ of the \GRC{} problem can be solved in polynomial time if $w(\call) = 2$.
\end{lemma}

\begin{proof}
    Given the instance $I$, we apply the aforementioned method to fixed edges to produce an equivalent instance $I'=(\texttt{d}', \call')$ containing only forbidden edges.
    The problem then reduces to finding a subgraph of the possibility graph $\calg$ of $I'$ that realizes the degree list~$\texttt{d}'$.
    By interpreting $\texttt{d}'$ as a function $f \colon V \to \N$, the problem becomes finding an $f$-factor of $\calg$, which is solvable in cubic time using, for example, the algorithm of Anstee \cite{Ans85}. \qed
\end{proof}

Interestingly, the \GRC{} problem remains solvable in polynomial time even when cuts of size three are present.
This is because cuts of size three actually have no more restraining power on realizability than cuts of size two, in the sense that we can construct an equivalent instance containing only cuts of size at most two that maintain the same realizability as the original instance.

{

\def \scaling {0.7}

\begin{figure}[!b]

\begin{subfigure}[b]{0.49\textwidth}
\centering       
\begin{tikzpicture}[scale=\scaling, transform shape]
    \sample
    \node[above, black] at (- 1.3, 0) {$d(S)$};
     \foreach \x/\y/\z in {H1/T111/H2, H3/T121/H4, H5/T122/H6, H7/T211/H8, H9/T221/H10, H11/T222/H12} 
       \draw[/edgeType4] (\x) -- (\y) -- (\z);     
\end{tikzpicture}
\caption{$\ell = d(S)$}
\label{fig:proof_w=3_l=dS}
\end{subfigure}    
%
%
\begin{subfigure}[b]{0.49\textwidth}
\centering
\begin{tikzpicture}[scale=\scaling, transform shape]
    \sample
    \node[above, black] at (- 1.3, 0) {$d(S) - 6$};
    \draw[/edgeType4] (T111) -- (T121) -- (T122) -- (T111);
    \draw[/edgeType1] (T211) -- (T221) -- (T222) -- (T211); 
    \draw[/edgeType4] (T211) -- (T221) -- (T222) -- (T211); 
\end{tikzpicture}
\caption{$\ell = d(S)-6$}
\label{fig:proof_w=3_l=dS-6}
\end{subfigure}
%
%
\end{figure} 
\begin{figure}[!ht]\ContinuedFloat  
\centering     

\begin{subfigure}[b]{0.49\textwidth}
\centering  
\begin{tikzpicture}[scale=\scaling, transform shape]
    \sample
    \node[above, black] at (- 1.3, 0) {$d(S) - 2$};
    \node[/nodeType, draw] (X) at (6, .5) {$x$};
    \foreach \x/\y in {X/T211, X/T221, X/T222} 
       \draw[/edgeType1] (\x) -- (\y); 
    
    \foreach \x/\y in {H1/T111, H2/T111, H3/T121, H5/T122, T121/T122, H7/T211, H8/T211, H9/T221, T221/X, X/T222, H11/T222} 
       \draw[/edgeType4] (\x) -- (\y);      
\end{tikzpicture}
\caption{$\ell = d(S) - 2$}
\label{fig:proof_w=3_l=dS-2}
\end{subfigure}
%
%
\begin{subfigure}[b]{0.49\textwidth}
\centering
\begin{tikzpicture}[scale=\scaling, transform shape]
    \sample
    \node[above, black] at (- 1.3, 0) {$d(S) - 4$};
    \node[/nodeType, draw] (X) at (6, .9) {$x$};
    \node[/nodeType, draw] (Y) at (2, -.3) {$y$};
    \foreach \x/\y in {X/T211, X/T221, X/T222, Y/T211, Y/T221, Y/T222} 
       \draw[/edgeType1] (\x) -- (\y);

    \foreach \x/\y/\z in {H2/T111/T121, T121/T122/H5, H8/T211/X, X/T221/Y, X/T222/H11} 
       \draw[/edgeType4] (\x) -- (\y) -- (\z);
\end{tikzpicture}
\caption{$\ell = d(S) - 4$}
\label{fig:proof_w=3_l=dS-4}
\end{subfigure}

\caption{
    Illustration of all cases for a cut $(S, \ell)$ with $S = \set{u, v, w}$, assuming $d_u = d_v = d_w = 2$ (so $d(S) = 6$). Solid edges represent possible edges, dashed edges are forbidden, and blue-highlighted edges belong to a realization. In each case, the left image shows a realization satisfying $(S, \ell)$, while the right image shows the equivalent realization of the modified instance without the cut.
}

\end{figure}
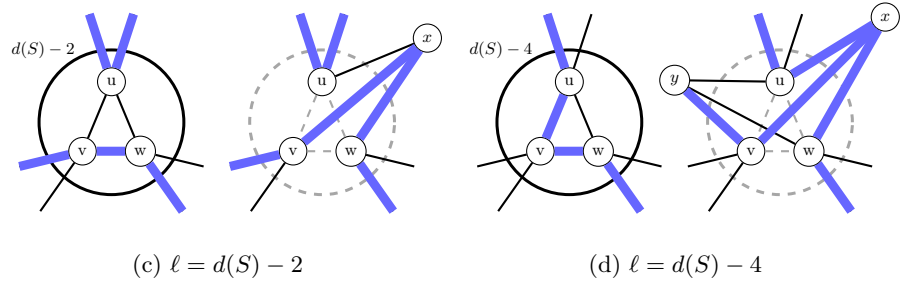
}

\begin{theorem}
\label{thm:size3}
    Any instance $I=(\texttt{d}, \call)$ of the \GRC{} problem can be solved in polynomial time if $w(\call) = 3$.
\end{theorem}

\begin{proof}
    We will show that it is possible to construct, in polynomial time, an instance $(\texttt{d}', \call')$ such that $w(\call') = 2$ and $(\texttt{d}', \call')$ is realizable if and only if $(\texttt{d}, \call)$ is realizable.
    This will complete our proof by applying \cref{thm:size2} to~$(\texttt{d}', \call')$.
    To achieve this, consider a cut $(S, \ell) \in \call$ where $S = \set{u, v, w}$.
    From \cref{thm:feasible_cut_sizes}, we know there are exactly four possible values for $\ell$: $d(S)$, $d(S) - 2$, $d(S) - 4$, and $d(S) - 6$.
    In each case, we show that $(S, \ell)$ can be replaced by cuts of size two, along with, possibly, some additional vertices. 
    Recall that forbidding or fixing an edge $uv$ is a constraint that we can express through a cut constraint $(\set{u, v}, d_u + d_v)$ or $(\set{u, v}, d_u + d_v - 2)$, respectively.

    Case 1: $\ell = d(S)$.
    In this case, all edges incident to $S$ must be included in the edge cut $\partial(S)$. So, this cut effectively forbids the edges $uv$, $uw$, and $vw$, as shows \cref{fig:proof_w=3_l=dS}.

    Case 2: $\ell = d(S) - 6$.
    This case is similar to Case 1, but we require here all three edges between vertices in $S$ to be present. Therefore, we fix the edges $uv$, $uw$, and $vw$, as in \cref{fig:proof_w=3_l=dS-6}.

    Case 3: $\ell = d(S) - 2$.
    This cut enforces that exactly one edge within $S$ must be included in any realization. Equivalently, this constraint requires selecting two vertices from $S$ to decrease their degrees by $1$ each.

    To eliminate this cut from $\call$ (see \cref{fig:proof_w=3_l=dS-2}), we proceed as follows.
    We create a new vertex $x$, set $d_x = 2$, and forbid all edges between $x$ and vertices outside $S$.
    Additionally, we forbid the edges between the vertices within $S$.
    In this setting, the two vertices in $S$ adjacent to $x$ will simulate the selection of an edge in a realization of the original instance.
    Assume, without loss of generality, that a realization $G$ of the original instance exists with $vw \in E(G)$. Then, in the modified instance, a realization $G'$ exists in which $x$ is adjacent to both $v$ and $w$ and $vw$ is not present. The converse also holds, ensuring that this modification to $\call$ preserves the realizability of the instance.
    
    Case 4: $\ell = d(S) - 4$.
    In this case, exactly two edges within $S$ must be included in any realization.
    Following the same rationale as in the previous case, this amounts to the degrees of two vertices in $S$ being reduced by $1$, while the degree of the remaining vertex is reduced by $2$. Note that since we only have three vertices and, therefore, three possible edges, the choice of the two edges can be defined by selecting which vertex of $S$ will have its degree decreased by 2.

    This can be equivalently accomplished by proceeding as follows (see \cref{fig:proof_w=3_l=dS-4}).
    We create two new vertices, $x$ and $y$, and set $d_x = 3$ and $d_y = 1$.
    We fix all three edges from $x$ to~$S$, forbid all edges between $y$ and vertices outside~$S$, and forbid the edges within $S$. Note that the fixed edges from $x$ to $S$ reduce the degree of each vertex in $S$ by $1$, while the vertex in $S$ that connects to $y$ will have its degree reduced by an additional $1$, simulating the required decrease~of~$2$.
    Therefore, without loss of generality, there is a realization $G$ of the original instance such that $uv, vw \in E(G)$ if and only if there is a realization $G'$ of the modified instance with $xu, xv, xw, yv \in E(G')$ and $uv, vw \notin E(G')$.

    \paragraph{}
    We apply these modification rules to each cut $(S, \ell)$ of size three in $\call$, resulting in a new instance $I'=(\texttt{d}', \call')$ with $w(\call') = 2$ and the same realizability as $(\texttt{d}, \call)$. In Cases 1 and 2, each cut $(S, \ell)$ is replaced by three smaller cuts, while in Cases 3 and 4, $\calo(n)$ additional cuts are required. Nevertheless, the total size of $\call'$ and the number of vertices are only increased polynomially. Therefore, by applying \cref{thm:size2} to $I'$, we solve our original instance in polynomial time.
    \qed
\end{proof}

\section{Large Cuts}
\label{sec:large_cuts}

Now we discuss the \GRC{} with $w(\call) \ge 4$. Interestingly, we get a dichotomy and can no longer solve \GRC{} within polynomial time unless $\classP=\classNP$. Our hardness result holds even if \texttt{d} is a sequence of ones. Additionally, we explore restrictions over the possibility graph $\calg$ and show that \GRC{} is \classNPC{} for $w(\call) = 6$ even if $\calg$ is bipartite and subcubic. In contrast, if $\calg$ is a tree, we argue how the problem can be efficiently solved.

\subsection{Cuts of Size Four}

Regarding the size of cuts, one might initially think that the approach of \cref{thm:size3}, which reduces an instance with $w(\call) = 3$ to one with $w(\call) = 2$, could be extended to larger cuts.
However, this extension is not feasible. For cuts $(S, \ell)$ of size three, the number of edges within $S$ uniquely determines how much the degrees of each vertex are affected. In contrast, when $\size{S} = 4$, this property already does not hold.
Consider, for instance, a cut $(S, \ell)$ where $S = \set{u, v, w, x}$ and $\ell = d(S) - 4$. In any realization, there must be exactly two edges within $S$.
If, in a realization $G$, these edges are disjoint (e.g., $uv, wx \in E(G)$), then each vertex in $S$ has its degree decreased by $1$.
On the other hand, the edges might not be disjoint (e.g., $uv, uw \in E(G)$). In this case, one vertex has its degree decreased~by~$2$, two vertices have a reduction of $1$, and one vertex remains unchanged.
Since it is impossible to determine beforehand which of these configurations applies, the reduction strategy used in Theorem~\ref{thm:size3} cannot be generalized.

In fact, we show that when $w(\call) = 4$, the \GRC{} problem cannot be solved in polynomial time unless $\classP{} = \classNP{}$.
We use a restricted variant of the {$1$-in-$3$-SAT} problem \cite{Ga79} in our proof.
%
%
%
In our case, we consider propositional formulas in conjunctive normal form where every variable appears exactly three times, two times as a positive literal (i.e., not negated) and one time as a negative literal (i.e., negated). We ask if it is possible to find a satisfying assignment such that each clause has exactly one literal that evaluates to true while the rest are false.
Additionally, we require that each clause has two or three literals (it is trivial to handle clauses with only one literal, so we assume they are preprocessed away).
We call this variant \rXthreeSAT{}, and although pretty restricted, this problem remains \classNPC{}.


\defproblema{
\rXthreeSAT{}}
{
A set of variables $X$ and a formula~$\phi$ in conjunctive normal form over $X$ such that:
\begin{itemize}
    \item each variable of $X$ occurs twice as a positive literal and once as a negative literal;
    \item each clause of $\phi$ has two or three literals.
\end{itemize}
}
{
Is there a truth assignment of $X$ such that exactly one literal in every clause of $\phi$ is true?
}

\begin{lemma}
\label{thm:2-1-xsat}
    \rXthreeSAT{} is \classNPC{}.
\end{lemma}

\begin{proof}
    Notice that the problem is in \classNP{}. To prove that it is \classNPH{}, we show a reduction from Positive $1$-in-$3$-SAT, the monotone version of $1$-in-$3$-SAT in which all literals are positive. This problem is known as \classNPC{}~\cite{Ga79}.

    Let~($X$, $\phi$) be an input of the Positive $1$-in-$3$-SAT problem.
    We will show how to add clauses and variables to $(X, \phi)$ to obtain an equivalent instance $(X', \phi')$ of \rXthreeSAT{}. We start with $X' = X$ and $\phi' = \phi$.
    Let $x \in X$. First, consider that $x$ only occurs once in $\phi$ (we know it appears as a positive literal). Then we add to $\phi'$ the redundant clause $(x + \neg{x})$. Now $x$ occurs in $\phi'$ twice as a positive literal and once as a negative one, and $\phi'$ is equivalent to $\phi$.

    Now suppose $x$ has two appearances in $\phi$. Then we create a new variable $a$; we add the clause $(x + \neg{a})$ to set a logical equivalence between $x$ and $a$ (since exactly one of $x$ and $\neg{a}$ must be true in a satisfying truth assignment, $x \equiv a$ in any feasible assignment).
    This allows us to replace the second occurrence of $x$ with its equivalent variable~$a$. The resulting $\phi'$ is equivalent to $\phi$, and by also adding the redundant clause $(a + \neg{x})$ to $\phi$, both $x$ and the new variable $a$ occur twice as a positive literal and once as a negative one.

    Finally, we generalize this last idea. Lets say that $x$ occurs $t$ times, $t \geq 2$. We create $t-1$ variables, $a_1, a_2, \ldots, a_{t-1}$, which will be all equivalent to $x$. To this end, we add $t$ clauses to $\phi'$: $(x + \neg{a}_1), (a_1 + \neg{a}_2), \ldots, (a_{t-2} + \neg{a}_{t-1}), (a_{t-1} + \neg{x})$. Now $x$ and the new variables $a_1, \ldots, a_{t-1}$ are all equivalent, i.e., they must have the same truth value in any assignment that satisfies exactly one literal of every clause. In $\phi'$, we maintain the first appearance of $x$, but the second one is replaced by $a_1$, the third is replaced by $a_2$, and so forth. The resulting $\phi'$ is still equivalent to the original $\phi$, and if we do this for every variable in $X$, we obtain an instance $(X', \phi')$ that is an instance of \rXthreeSAT{}. Furthermore, we remark that in $\phi'$, every clause that comes from~$\phi$ has three literals, and every clause that we created for the reduction has two literals; thus, $\phi'$ only has clauses of sizes two and three.
    \qed
\end{proof}

To argue the hardness of \GRC{} with $w(\call) = 4$, it will be useful to know beforehand how many variables in a \rXthreeSAT{} instance must be set to true.
To this end, we consider a more restricted problem, which remains hard.


\defproblema{
\kXthreeSAT{}}
{
A tuple $(X, \phi, k)$, where $(X, \phi)$ is an instance of \rXthreeSAT{} and $k$ is a nonnegative integer.
}
{
Is there a feasible solution to $(X, \phi)$ in which exactly $k$ variables are assigned to true?
}

\begin{lemma}
\label{thm:k-true-xsat}
    \kXthreeSAT{} cannot be solved in polynomial time unless $\classP{} = \classNP{}$.
\end{lemma}

\begin{proof}
    Follows directly from Lemma~\ref{thm:2-1-xsat}. Suppose that $\cala$ is an algorithm that decides \kXthreeSAT{} in polynomial time. Given an instance $(X, \phi)$ of \rXthreeSAT{} with $n$ variables, we run $\cala$ on inputs $(X, \phi, 0), \ldots, (X, \phi, n)$. If $\cala$ accepts any of these, we determine that the answer to $\phi$ is YES; otherwise, it is NO. Therefore, we can solve \rXthreeSAT{} in polynomial time, implying that $\classP{} = \classNP{}$.
    \qed
\end{proof}
\medskip


Equipped with the \kXthreeSAT{} problem and by knowing that it cannot be solved in polynomial time unless $\classP{} = \classNP{}$,
we can proceed to show the hardness of the \GRC{} problem when $w(\call) = 4$.

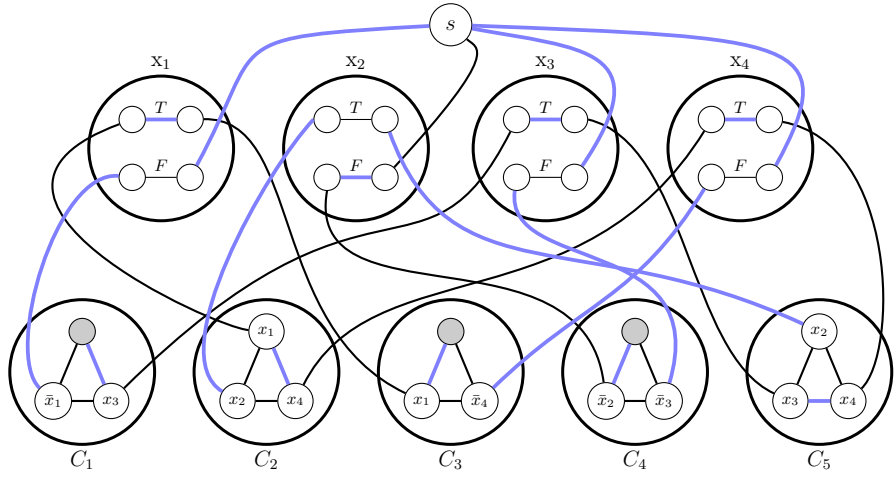
\begin{figure}[!ht]
    \centering
    \begin{tikzpicture}[scale=0.7, transform shape]
    \pgfkeys{/edgeType1/.style={solid, thick},
            /edgeType4/.style={blue!50,line width=0.5mm, draw opacity=0.7}}
    \node[draw, circle, minimum size=.8cm] (n) at (5.5, 1.25) {\large$s$}; 
    \squareGadget{1}{0}{0}{$x$}
    \squareGadget{2}{3.7}{0}{$x$}
    \squareGadget{3}{7.3}{0}{$x$}
    \squareGadget{4}{11}{0}{$x$}
    
    \pairGadget{1}{-1.5}{-4.25}{}{\large$C_1$}{$\bar{x}_1$}{$x_3$}
    \triangleGadget{1}{2}{-4.25}{}{\large$C_2$}{$x_1$}{$x_2$}{$x_4$}
    \draw[/edgeType1] (T111) -- (T121) -- (T122) -- (T111);
    \pairGadget{2}{5.5}{-4.25}{}{\large$C_3$}{$x_1$}{$\bar{x}_4$}
    \pairGadget{3}{9}{-4.25}{}{\large$C_4$}{$\bar{x}_2$}{$\bar{x}_3$}
    \triangleGadget{2}{12.5}{-4.25}{}{\large$C_5$}{$x_2$}{$x_3$}{$x_4$}
    \draw[/edgeType1] (T211) -- (T221) -- (T222) -- (T211);

    \draw[/edgeType1] (n) .. controls (1, 1) and (2, 1) .. (S122);
    \draw[/edgeType1] (n) .. controls (6, .8) and (6.5, .8) .. (S222);
    \draw[/edgeType1] (n) .. controls (8, .8) and (9.5, .8) .. (S322);
    \draw[/edgeType1] (n) .. controls (12, 1) and (13, 1) .. (S422);
    \begin{pgfinterruptboundingbox}
    \draw[/edgeType1] (S111) .. controls (-4.5, -2) and (1,-4.5) .. (T111);   
    \end{pgfinterruptboundingbox}
    \draw[/edgeType1] (S112) .. controls (2.7, -.5) and (2,-4.5) .. (P211);
    \draw[/edgeType1] (S121) .. controls (-2, -1.5) and (-3,-5) .. (P111);
    \draw[/edgeType1] (S211.west) .. controls (2, -1.5) and (0, -5) .. (T121);
    \draw[/edgeType1] (S212) .. controls (5.5, -4) and (8, -2.5) .. (T211);
    \draw[/edgeType1] (S221.south) .. controls (2.7, -4.5) and (8, -2.5) .. (P311);
    \draw[/edgeType1] (S311) .. controls (5, -4) and (4, -1) .. (P112);
    \draw[/edgeType1] (S312.east) .. controls (10, -1) and (10, -5) .. (T221);
    \draw[/edgeType1] (S321) .. controls (6.5, -3.5) and (10.5, -2.7) .. (P312);
    \draw[/edgeType1] (S411) .. controls (8, -4.5) and (4, -3) .. (T122);
    \draw[/edgeType1] (S412) .. controls (14, -1) and (14, -5) .. (T222);
    \draw[/edgeType1] (S421) .. controls (9, -4) and (8, -4) .. (P212);

    \foreach \i in{1,3,4}
         \draw[/edgeType4] (S\i11) -- (S\i12);

    \draw[/edgeType4] (S221) -- (S222);
    \draw[/edgeType4] (P101) -- (P112);
    \draw[/edgeType4] (P201) -- (P211);
    \draw[/edgeType4] (P301) -- (P311);
    \draw[/edgeType4] (T111) -- (T122);
    \draw[/edgeType4] (T221) -- (T222);
    
    \draw[/edgeType4] (S211.west) .. controls (2, -1.5) and (0, -5) .. (T121);
    \draw[/edgeType4] (S212) .. controls (5.5, -4) and (8, -2.5) .. (T211);
    \draw[/edgeType4] (n) .. controls (1, 1) and (2, 1) .. (S122);
    \draw[/edgeType4] (n) .. controls (8, .8) and (9.5, .8) .. (S322);
    \draw[/edgeType4] (n) .. controls (12, 1) and (13, 1) .. (S422);
    \draw[/edgeType4] (S121) .. controls (-2, -1.5) and (-3,-5) .. (P111);
    \draw[/edgeType4] (S321) .. controls (6.5, -3.5) and (10.5, -2.7) .. (P312);
    \draw[/edgeType4] (S421) .. controls (9, -4) and (8, -4) .. (P212);    
\end{tikzpicture}
    \caption{
    Illustration of the possibility graph $\calg$ built from an instance $(X, \phi, k)$ of \kXthreeSAT{}{} with
    $X = \set{x_1, x_2, x_3, x_4}$, 
    $\phi = (\bar{x}_1 \oper x_3)(x_1 \oper x_2 \oper x_4)(x_1 \oper \bar{x}_4)(\bar{x}_2 \oper \bar{x}_3)(x_2 \oper x_3 \oper x_4)$
    and $k = 1$.
    Gray vertices represent artificial vertices created for clauses with only two literals.
    The highlighted edges show an example of a feasible realization for such an instance.
    }
    \label{fig:proof_w=4}
    \vspace{-.5cm}
\end{figure}

\begin{theorem}
    The \GRC{} problem cannot be solved in polynomial time unless $\classP = \classNP$ even when $w(\call) = 4$ and all degrees in the degree sequence $\texttt{d}$ are 1.
\end{theorem}
\begin{proof}


    We present a reduction from the $k$-True \rXthreeSAT{} to the \GRC{} problem with the desired properties.
    To this end, let $(X, \phi, k)$ be an instance of the $k$-True \rXthreeSAT{}.
    We now describe the building of the instance $I = (\texttt{d}, \call)$ of \GRC{}.
    Refer to \cref{fig:proof_w=4} for an illustrative example.
    We start with $\texttt{d}$ and $\call$ empty, and we let $V$ be the corresponding set of vertices and $\calg$ be its possibility graph.
    First, for each variable $x_i$ in $X$, we build a \textit{variable gadget} as follows.
    Let $X_i$ be a set of four vertices, namely $X_i = \set{ x^i_{T_1}, x^i_{T_2}, x^i_{F_1}, x^i_{F_2}}$.
    We add $X_i$ to $V$ and set $d_u = 1$ for each $u \in X_i$, we add the edges $x^i_{T_1} x^i_{T_2}$ and $x^i_{F_1} x^i_{F_2}$ to $\calg$, and we add the cut $( X_i, 2)$ to $\call$.
    Moreover, for each clause $C_j$ of $\phi$, we define its \textit{clause gadget}.
    We create a new vertex for each literal that occurs in $C_j$. If $C_j$ has only two literals, we create another artificial one. Let $Y_j$ be this set of three vertices.
    We add $Y_j$ to $V$ and set $d_v = 1$ for each $v \in Y_j$, we add all the edges between vertices of $Y_j$ to $\calg$, and we add the cut $(Y_j, 1)$ to $\call$.

    To conclude the definition of the vertices $V$,
    we create a vertex $s$ and set $d_s = n - k$.
    The set $V$ of our instance is thus composed of $s$ and the vertices of each variable and clause gadget.
    To finish $\calg$'s construction, we join the vertex and clause gadgets as follows.
    For each variable $x_i$, let $C_{i_1}$ and $C_{i_2}$ be the two clauses where $x_i$ appears as a positive literal, and $C_{i_3}$ the clause in which it appears as a negative literal.
    We connect $x^i_{T_1}$, $x^i_{T_2}$ and $x^i_{F_1}$ to the vertex that corresponds to its literal in $C_{i_1}$, $C_{i_2}$ and $C_{i_3}$, respectively, while $x^i_{F_2}$ is connected to $s$.
    The final cut list $\call$ is given by the aforementioned cuts in the vertex and clause gadgets, plus the ones defining $\calg$.
    Now we show that there is a realization for such $(\texttt{d}, \call)$ if and only if $(X, \phi)$ is satisfiable using exactly $k$ variables as true.

    Let $\hat{x}$ be a feasible solution to $(X, \phi)$ using exactly $k$ variables as true.
    Let $G$ be a spanning subgraph of $\calg$, initially with no edges.
    For each $\hat{x}_i$ from $\hat{x}$, if $\hat{x}_i = T$, we add to $G$ the edge $x^i_{F_1}x^i_{F_2}$ along with the edges from $x^i_{T_1}$ and~$x^i_{T_2}$ to their corresponding positive literals in the clause gadgets.
    Similarly, if $\hat{x}_i = F$, then we add to $G$ the edge $x^i_{T_1}x^i_{T_2}$ along with the edges from $x^i_{F_1}$ to its corresponding negative literal in the clause gadget and from $x^i_{F_2}$ to $s$.
    At last, for each clause $C_j$ in $\phi$, let $\hat{x}_{j_1}$ and $\hat{x}_{j_2}$ be the corresponding vertices of the two literals in $C_j$ that were evaluated as false in $\hat{x}$ in case $C_j$ has three literals, or the literal evaluated to false and the artificial vertex added in the clause gadget of $C_j$, otherwise. We~then add the edge $\hat{x}_{j_1}\hat{x}_{j_2}$ to $G$ for each clause $C_j$.
    In both cases, in $G$ we have that $\partial(X_i) = 2$. 
    Since $\hat{x}_i$ is a feasible solution, in each clause $C_j$, exactly one variable is evaluated to true, which implies that $\partial(Y_j) = 1$ in $G$.
    Furthermore, the degree of all vertices except for $s$ is $1$, and since there are exactly $k$ variables in $\hat{x}$ assigned to true, there are $n-k$ edges in $G$ from vertices $x^i_{F_2}$ to $s$, thus respecting $d_s$. Therefore, $G$ is a realization of~$(\texttt{d}, \call)$.

    Conversely, let $G$ be a realization of $(\texttt{d}, \call)$.
    Since $(X_i,2)\in \call$ for each variable $x_i$, $d_u = 1$ for each $u \in X_i$, and $E(\calg[X_i])=\{ x^i_{T_1} x^i_{T_2}, x^i_{F_1} x^i_{F_2}\}$. It holds that either $x^i_{T_1}, x^i_{T_2}$ or $x^i_{F_1},x^i_{F_2}$ have neighbors outside $X_i$ in $G$. Therefore, we define an assignment $\hat{x}$ to $(X,\phi)$ as follows: $x_i=T$ if $x^i_{T_1}, x^i_{T_2}$ have neighbors outside $X_i$ in $G$, otherwise $x_i=F$. As $(Y_j, 1)\in \call$ for each clause $C_j$, it follows that each clause of $\phi$ has exactly one literal assigned to true in $\hat{x}$. Given that $d(s)=n-k$, the vertex $s$ has $n-k$ neighbors in $G$. By construction, each neighbor of $s$ in $G$ is a $x^i_{F_2}$ vertex for some $i$. Thus, $\hat{x}$ has exactly $n-k$ negative literals evaluating true, and therefore $\hat{x}$ is a feasible solution to $(X, \phi)$ in which exactly $k$ variables are assigned to true.
    
    To have all degrees in the degree sequence $\texttt{d}$ equal 1 is enough to modify the construction, replacing $s$ by $n-k$ copies each with desired degree equals one in $\texttt{d}$.  
    Thus, from this reduction, we conclude that if the \GRC{} problem is solvable in polynomial time, then we can also solve the $k$-True \rXthreeSAT{} problem in polynomial time, which would imply that $\classP{} = \classNP{}$ due to \cref{thm:k-true-xsat}. \qed
\end{proof}

\subsection{Restricted Possibility Graph}

We now move our attention to particular instances $(\texttt{d}, \call)$ of \GRC{} in which $w(\call)$ is not bounded, but the possibility graph $\calg$ belongs to a restricted graph class.
If $\calg$ is a tree, we can solve the \GRC{} in polynomial time.
For instance, we can construct a candidate realizing graph $G$ by processing $\calg$'s leaves iteratively, ensuring at each step that the degree constraints are met. If a violation occurs or the final vertex has a nonzero degree, we return NO; otherwise, we can verify whether $G$ realizes $\call$ in polynomial time.

\begin{proposition}
    \label{prep:tree_graph}
    Given an instance $(\texttt{d}, \call)$ of \GRC{} with a tree possibility graph $\calg$, we can decide if there is a solution in polynomial time.
\end{proposition}

\mycomment{
\begin{proof}
We begin with an empty graph $G = (V, \emptyset)$. If $\calg$ has at least two vertices, we select a leaf $v_i$. Let $v_j$ be the father of $v_i$. If the degree $d_i$ of $v_i$ is greater than one, no valid solution exists as there are not enough allowed edges, and we immediately return no. Otherwise, if $d_i = 1$, we decrease $d_j$ by one and add an edge between $v_i$ and $v_j$ in $G$. Then, we remove $v_i$ from $\calg$ along with its degree $d_i$ from the sequence \texttt{d}. These removals also happen if $d_i=0$.

We repeat this process, treating one leaf at a time until only a single vertex $v_i$ remains in $\calg$. At this final step, we return no if the remaining degree $d_i$ is not zero.
This approach guarantees that if the algorithm returns no at any point, the instance indeed has no solution. Otherwise, once the graph $G$ is constructed, we verify if $G$ realizes $\call$, outputting yes if it does and no otherwise, which can be done in polynomial time. \qed
\end{proof}
}

As it turns out, if we relax the restrictions on $\calg$ and allow a bipartite graph, we get a \classNPC{} problem.
To show this, we will make use of the 3-Dimensional Matching problem, which is defined next.


\defproblema{3-Dimensional Matching -- \TDM{}}
{
3 disjoint sets $X$, $Y$, and $Z$ with $|X| = |Y| = |Z| = n$, and a set of triples $T \subseteq X \times Y \times Z$.
}
{
Is there a subset $M \subseteq T$ such that $\size{M}=n$ and no two triples of $M$ intersect?
}

This problem remains \classNPC{} if no element occurs in more than three triples \cite{Ga79}. We refer to this particular case as \TDMT{}.

\begin{theorem}
    The \GRC{} problem is \classNPC{} when the possibility graph $\calg$ is subcubic and bipartite, even when $w(\call) = 6$ and \texttt{d} is a sequence of ones.
\end{theorem}
\begin{proof}
The \GRC{} is clearly in \classNP{}. We show a reduction from \TDMT{} to prove its hardness.
Consider an instance $(X, Y, Z, T)$ of \TDMT{} where $|X| = |Y| = |Z| = n$. Without loss of generality, assume that every element in $X$, $Y$, and $Z$ appears in at least one triple in $T$ (see Figure \ref{fig:3dm3red} for an illustration of the reduction).

To construct the vertex set $V$ for our \GRC{} instance, we proceed as follows: for each element $x_i \in X$, we create a corresponding vertex $x_i$ in $V$.
For each element $y_j \in Y$, we construct a group of vertices $V_j$, determined by the number of triples in $T$ containing $y_j$. If $y_j$ appears in $l$ triples, where $1 \leq l \leq 3$, we create $2l$ vertices labeled $y^j_{1,a}, \ldots, y^j_{l,a}$ and $y^j_{1,b}, \ldots, y^j_{l,b}$, and group these into two sets, $Y^j_a = \{y^j_{1,a}, \ldots, y^j_{l,a}\}$ and $Y^j_b = \{y^j_{1,b}, \ldots, y^j_{l,b}\}$. Define $Y_a = \bigcup_j Y^j_a$ and $Y_b = \bigcup_j Y^j_b$.
Lastly, for each element $z_k \in Z$, we create a vertex $z_k$ in $V$. In total, this construction yields $2(|T| + n)$ vertices, where $V = X \cup Y_a \cup Y_b \cup Z$.

We define the degree sequence \texttt{d} such that each vertex in $V$ has a degree exactly one. This degree constraint ensures that each vertex is matched with only one other vertex, guaranteeing that any feasible solution forms a matching.
%
Next, we construct the cut list $\call$. For each group $V_j$, we add the pair $(V_j, 2)$ to $\call$. This is the largest cut with a size of at most six, enforcing exactly two edges connecting vertices in $V_j$ to vertices outside of $V_j$. We will later argue that they specifically connect to a vertex of $X$ and a vertex of $Z$.

We then add cuts of size two, as per Remark \ref{thm:fixed_forbidden_edges}, to $\call$ to prohibit all edges except those allowed by the following rules. For each $y_j \in Y$, let $(x_{i_1}, y_j, z_{k_1}), \ldots,$ $(x_{i_l}, y_j, z_{k_l})$ denote the $l$ triples of $T$ in which $y_j$ appears. We only allow edges from $x_{i_u}$ to $y^j_{u,a}$, from $y^j_{u,a}$ to $y^j_{u,b}$, and from $y^j_{u,b}$ to $z_{k_u}$, for each $1 \leq u \leq l$.
The construction encodes the selection of a triple $(x_{i_v}, y_j, z_{k_v})$ by including the edges $x_{i_v} y^j_{v,a}$ and $y^j_{v,b} z_{k_v}$ in the realization, while the remaining vertices in $V_j$ forms a matching.
Observe that this instance's possibility graph $\calg$ is bipartite and subcubic. Each vertex of $Y_a\cup Y_b$ has degree two, while the vertices in $X \cup Z$ have a degree at most three, as no element occurs in more than three triples.

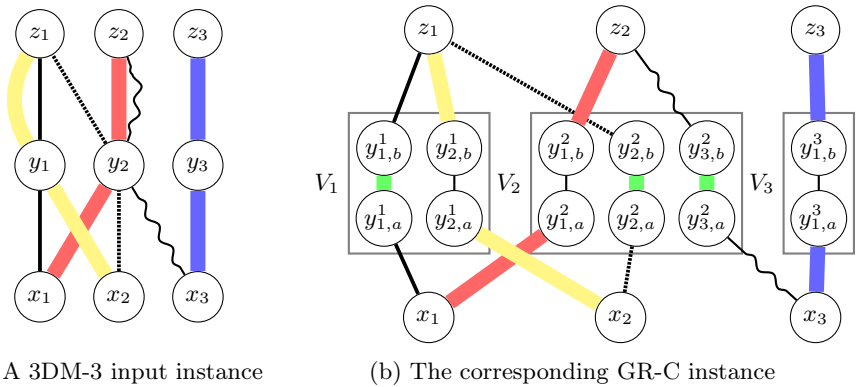
\begin{figure}[htb!]
    \centering
    \begin{subfigure}[b]{0.35\textwidth}
\centering
\raisebox{1\ht\strutbox}{
    \begin{tikzpicture}[scale=0.7, transform shape = false, rotate=90]
        \pgfkeys{/nodeType/.style={circle, draw},
        /edgeType1/.style={solid, thick, line width=.5mm},
        /edgeType2/.style={solid, thick},
        /edgeType3/.style={thick, decorate, decoration={snake, amplitude=.5mm,pre=lineto,pre length=3pt,post=lineto,post length=3pt}},
        /edgeType4/.style={line width=2mm, draw opacity=0.7}} 
    
        \def\n{3}
        \foreach \i in {1, ..., \n}{
            \node[/nodeType] (x\i) at (0, -1.5*\i) {$x_\i$}; 
            \node[/nodeType] (y\i) at (2.5, -1.5*\i) {$y_\i$}; 
            \node[/nodeType] (z\i) at (5, -1.5*\i) {$z_\i$};
        }
    
        \draw[/edgeType1] (x1) -- (y1) -- (z1);
        \draw[/edgeType2] (x1) -- (y2) -- (z2);
        \draw[/edgeType2] (x2) -- (y1) to[bend left=30] (z1);
        \draw[/edgeType1] [dotted, dash pattern=on 1pt off .5pt] (x2) -- (y2) -- (z1);
        \draw[/edgeType3] (x3) -- (y2);
        \draw[/edgeType3] (y2) to[bend right=20](z2);
        \draw[/edgeType2] (x3) -- (y3) -- (z3);
    
         \draw[/edgeType4] [red!60] (x1) -- (y2) -- (z2);
         \draw[/edgeType4] [yellow!60] (x2) -- (y1) to[bend left=30] (z1);
         \draw[/edgeType4] [blue!60] (x3) -- (y3) -- (z3);    
    \end{tikzpicture}
}
\caption{A \TDMT{} input instance}
\label{fig:sample3dm3}
\end{subfigure}
\quad
\begin{subfigure}[b]{0.58\textwidth}
\centering
\begin{tikzpicture}[scale=0.85, transform shape = false, rotate=90]
    \pgfkeys{/edgeType1/.style={solid, thick, line width=.5mm},
    /edgeType2/.style={solid, thick},
    /edgeType3/.style={thick, decorate, decoration={snake, amplitude=.5mm,pre=lineto,pre length=5pt,post=lineto,post length=5pt}},
    /edgeType4/.style={line width=2mm, green!60, draw opacity=0.7}} 

    \node[circle, draw] (x1) at (-1.0, -1.25) {$x_1$};
    \node[circle, draw] (z1) at (4-0.5, -1.25) {$z_1$};
    \setVi{1}{0}{0}{2}
    \node[circle, draw] (x2) at (-1.0, -4.25) {$x_2$};
    \node[circle, draw] (z2) at (4-0.5, -4.25) {$z_2$};
    \setVi{2}{0}{-2.85}{3}
    \node[circle, draw] (x3) at (-1.0, -7.3) {$x_3$};
    \node[circle, draw] (z3) at (4-0.5, -7.3) {$z_3$};
    \setVi{3}{0}{-6.8}{1}

    \draw[/edgeType1] (x1) -- (y111);
    \draw[/edgeType1] (y112) -- (z1);
    \draw[/edgeType2] (y111) -- (y112);
    \draw[/edgeType2] (x2) -- (y121) -- (y122) -- (z1);
    \draw[/edgeType2] (x1) -- (y211) -- (y212) -- (z2);
    \draw[/edgeType1] [dotted, dash pattern=on 1pt off .5pt] (x2) -- (y221);
    \draw[/edgeType2] (y221) -- (y222);
    \draw[/edgeType1] [dotted, dash pattern=on 1pt off .5pt] (y222) -- (z1);
    \draw[/edgeType3] (x3) -- (y231);
    \draw[/edgeType2] (y231) -- (y232);
    \draw[/edgeType3] (y232) -- (z2);
    \draw[/edgeType2] (x3) -- (y311) -- (y312) -- (z3);

    \draw[/edgeType4] [red!60] (x1) -- (y211);
    \draw[/edgeType4] [red!60] (y212) -- (z2);
    \draw[/edgeType4] (y221) -- (y222);
    \draw[/edgeType4] (y231) -- (y232);
    \draw[/edgeType4] [yellow!60] (x2) -- (y121);
    \draw[/edgeType4] [yellow!60] (y122) -- (z1);
    \draw[/edgeType4] (y111) -- (y112);
    \draw[/edgeType4] [blue!60] (x3) -- (y311);
    \draw[/edgeType4] [blue!60] (y312) -- (z3);
\end{tikzpicture}
\caption{The corresponding \GRC{} instance}
\label{fig:sampleGRC}
\end{subfigure}
    \caption{
    A \TDMT{} instance example where $T=\{(x_1, y_1, z_1), (x_1, y_2, z_2),\\ (x_2, y_1, z_1), (x_2, y_2, z_1), (x_3, y_2, z_2), (x_3, y_3, z_3)\}$. Distinct edge types are assigned to each triple, with a solution highlighted.
    The right image depicts the possibility graph of the reduced \GRC{} instance, with a feasible realization highlighted.
    }
    \label{fig:3dm3red}
    \vspace{-.5cm}
\end{figure}

If a feasible matching $M \subseteq T$ exists in the \TDMT{} instance, we can map it directly to the edges of a valid realization $G$ for the constructed \GRC{} instance.
For each $(x_{i_u}, y_j, z_{k_u}) \in M$, using the $u$th occurrence of $y_j$, we add the edges $x_{i_u} y^j_{u,a}$ and $y^j_{u,b} z_{k_u}$ to $G$. For each $(x_{i_v}, y_j, z_{k_v}) \in T \setminus M$, we add the edge $y^j_{v,a} y^j_{v,b}$.
Since $M$ is a solution, each vertex in $X \cup Z$ has degree one, satisfying the degree constraints.
Additionally, within each group $V_j$, exactly two vertices—$y^j_{u,a}$ and $y^j_{u,b}$ from a triple in $M$—connect to vertices in $X$ and $Z$, respectively. All other vertices within $V_j$ correspond to triples not included in $M$, forming a matching within $V_j$.
Therefore, $G$ fulfills both the degree sequence \texttt{d} by assigning degree one to every vertex and the cut list $\mathcal{L}$, meeting all the required constraints for a valid realization.

Conversely, if a graph $G$ exists that realizes both the degree sequence \texttt{d} and the cut list $\call$, we can construct a feasible matching $M \subseteq T$ for the \TDMT{} instance.
Since \texttt{d} specifies a degree of one for each vertex, the edges of $G$ form a matching.
Additionally, exactly two vertices within each group $V_j$ are matched to vertices of $X\cup Z$, meaning the remaining vertices within each $V_j$ form an internal matching.
These two externally matched vertices must correspond to the same triple in $T$; otherwise, the remaining vertices in $V_j$ could not be paired and meet the type $(V_j, 2)$ cut constraint.
Let $M$ consist of the triples in $T$ for which the associated $y^j_{u,a}$ and $y^j_{u,b}$ vertices in $V_j$ are connected to vertices in $X$ and $Z$, respectively.
Thus, by construction, vertex $x_i$ connects to $y^j_{u,a}$ and $y^j_{u,b}$ to $z_k$ if and only if the triple $(x_i, y_j, z_k)$ of $T$ belongs to $M$.
So $M$ contains exactly one triple per $V_j$, covering each element of $Y$ exactly once, thus $|M| = n$.
Since $G$ realizes $\call$, no edges exist between vertices in $X$ and $Z$. Hence, given that $G$ is a matching, each vertex in $X$ connects to exactly one vertex in $Y_a$, and each vertex in $Z$ connects to exactly one vertex in $Y_b$.
Consequently, $M$ constitutes a valid matching for the \TDMT{} instance.
\qed
\end{proof}

\section{Final Remarks}
\label{sec:final_remarks}
We introduced the \textsc{Graph Realization with Cut Constraints} problem in this work. 
This problem is interesting because it combines different graph theory concepts, including degree sequence, cut constraints, $f$-factors, and graph realization.
We provide a detailed characterization of its computational complexity based on the size of the cuts.
Our results show that the problem can be solved in polynomial time when the cuts are small enough (size at most three). However, the complexity significantly increases when the cuts are larger, and we proved that it becomes \classNPH{}.
An interesting direction for future work is identifying other graph classes where the possibility graph $\calg$ of a \GRC{} instance ensures polynomial-time solvability. For example, the idea of \cref{prep:tree_graph} might extend to cactus or, more generally, to graphs with bounded degeneracy or treewidth. The case of a planar possibility graph also deserves further investigation. 
We also ask about the complexity of {1-in-3 SAT}$_{(2,2)}$, the variant of {1-in-3 SAT} where each variable occurs exactly four times, twice positive and twice negative.

\begin{credits}
\subsubsection{\ackname}
This work was started during the 6th edition of WoPOCA, which took place in Campinas, São Paulo, Brazil. We thank the organizers and the agencies CNPq (process number 404315/2023-2) and FAEPEX (process number 2422/23).
We thank Esther Arkin, Soumya Banerjee, Rezaul Alam Chowdhury, Mayank Goswami, Dominik Kempa, Joseph Mitchell, Valentin Polishchuk, and Steven Skiena for some discussions prior the event, which helped motivate this work.
This research has received funding from Rio de Janeiro Research Support Foundation (FAPERJ) under grant agreement E-26/201.344/2021, the National Council for Scientific and Technological Development (CNPq) under grant agreements 309832/2020-9 and \mbox{163645/2021-3}, the São Paulo Research Foundation (FAPESP) under grant agreement 2022/13435-4, and the Brazilian Federal Agency for Support and Evaluation (CAPES) with process numbers 88887.646008/2021-00 and 88887.647870/2021-00. 

\subsubsection{\discintname}
The authors have no competing interests to declare that are
relevant to the content of this article.
\end{credits}

%
%
%
\bibliographystyle{splncs04}
\bibliography{references}

\end{document}